\documentclass[preprint,12pt]{elsarticle}

\usepackage{latexsym,amsfonts,amsmath,amssymb,mathrsfs,url,amsthm}
\usepackage{color}
\usepackage{graphicx}

\newtheorem{theorem}{Theorem}
\newtheorem{lemma}{Lemma}

\newtheorem{algorithm}{Algorithm}
\newtheorem{remark}{Remark}

\usepackage{lineno}

\newcommand{\rd}{\, \mathrm{d}}
\newcommand{\EE}{\mathbb{E}}

\newcommand{\VV}{\mathbb{V}}
\newcommand{\RR}{\mathbb{R}}

\newcommand{\bspi}{\boldsymbol{\pi}}
\newcommand{\bsx}{\boldsymbol{x}}
\newcommand{\bsy}{\boldsymbol{y}}
\newcommand{\bsz}{\boldsymbol{z}}

\newcommand{\Fcal}{\mathcal{F}}
\newcommand{\Ucal}{\mathcal{U}}

\journal{}

\begin{document}

\begin{frontmatter}



\title{A simple algorithm for global sensitivity analysis with Shapley effects}

\author{Takashi Goda}

\address{School of Engineering, University of Tokyo, 7-3-1 Hongo, Bunkyo-ku, Tokyo 113-8656, Japan.}

\begin{abstract}
Global sensitivity analysis aims at measuring the relative importance of different variables or groups of variables for the variability of a quantity of interest. Among several sensitivity indices, so-called Shapley effects have recently gained popularity mainly because the Shapley effects for all the individual variables are summed up to the overall variance, which gives a better interpretability than the classical sensitivity indices called main effects and total effects. In this paper, assuming that all the input variables are independent, we introduce a quite simple Monte Carlo algorithm to estimate the Shapley effects for all the individual variables simultaneously, which drastically simplifies the existing algorithms proposed in the literature. We present a short Matlab implementation of our algorithm and show some numerical results. A possible extension to the case where the input variables are dependent is also discussed.
\end{abstract}



\end{frontmatter}


\section{Introduction}\label{sec:intro}
Global sensitivity analysis provides an indispensable framework in measuring the relative importance of different variables or groups of variables for the variability of a quantity of interest \cite{GSAPrimer,IL15}, and has long been considered one of the central problems in reliability engineering \cite{HS96,BATS03,HJSS06}. In particular, since his pioneering work by Sobol' \cite{Sobol93,Sobol01}, variance-based sensitivity analysis has been studied intensively and applied to a variety of subjects in science and engineering. Two classical but still major sensitivity indices are called \emph{main effect} and \emph{total effect}. The main effect, also called the first-order effect, measures the variance explained only by an input variable or a group of input variables, whereas the total effect is given by the overall variance minus the variance explained only by the complement variables. When looking at these effects for the individual variables, the sum of the main effects is always less than or equal to the overall variance, while the sum of the total effects is always larger than or equal to the overall variance. In this way the classical sensitivity indices have a difficulty in normalization, and may cause some trouble when judging whether one input variable is more important than another.

Recently the connection between variance-based global sensitivity analysis and Shapley value from game theory has been studied by Owen \cite{Owen14}. The resulting sensitivity index defined for individual variables is called \emph{Shapley effect}. The Shapley effect takes its value between the main effect and the total effect, and importantly, the sum of the Shapley effects over all the individual variables is exactly equal to the overall variance. This normalization property gives us a better interpretability in determining whether one input variable is more important than another. Moreover, the Shapley effect has been proven to remove the conceptual problem inherit to both the main and total effects when the input variables are correlated \cite{SNS16,OP17}.

In this paper we mainly focus on the case where the input variables are independent and study Monte Carlo algorithms to estimate the Shapley effects efficiently. Building upon the previous works by Song et al.\ \cite{SNS16} and Broto et al.\ \cite{BBD20}, which themselves build upon the idea from \cite{CGT09}, in conjunction with the famous pick-freeze scheme \cite{Sobol01,Saltelli02,SAACRT10,Owen13,JKLNP14,Goda17}, we introduce a quite simple Monte Carlo algorithm to estimate the Shapley effects for all the individual variables simultaneously at the cost of $(d+1)N$, where $d$ denotes the number of variables and $N$ denotes the sample size. Our algorithm offers the following advantages:
\begin{enumerate}
\item Each Shapley effect is estimated unbiasedly.
\item The variance of the estimator is also estimated unbiasedly and decays at the canonical $1/N$ rate under some assumption.
\item The above advantages lead to an approximate confidence interval for the Shapley effect without requiring independent Monte Carlo trials or bootstrap resampling.
\item The sum of the estimates for all the Shapley effects is an unbiased estimator of the overall variance.
\end{enumerate}
We also discuss how to extend our algorithm to the case where the input variables are dependent. Although we are particularly interested in Monte Carlo sampling-based approaches in this paper, note that there are numerous different approaches to  carry out a global sensitivity analysis proposed in the literature, such as random balance designs \cite{TGM06}, polynomial chaos expansion \cite{S08}, fast Fourier transform \cite{P10}, low-rank tensor approximation \cite{KS16} and random forests \cite{ALP21}. We refer the reader to \cite{PT17} for a recent review on various estimation algorithms for classical sensitivity indices.

The rest of this paper is organized as follows. In Section~\ref{sec:gsa_review}, we give an overview of variance-based global sensitivity analysis and introduce several sensitivity indices. In Section~\ref{sec:mc_estimate} we present our simple Monte Carlo algorithm to estimate the Shapley effects for all the input variables simultaneously. We also show some good properties of our algorithm as explained above. Numerical experiments in Section~\ref{sec:numerics} confirm the effectiveness of our algorithm for three test cases. We conclude the paper in Section~\ref{sec:extension} with discussing two possible extensions of our algorithm, one for an additional cost saving and the other for the case where the input variables are dependent. In \ref{app:matlab}, we provide a short Matlab implementation of our algorithm, which we use for one test case.

\section{Variance-based sensitivity analysis}\label{sec:gsa_review}

\subsection{ANOVA decomposition}\label{subsec:anova}
Let $\bsx=(x_1,\ldots,x_d)\in \Omega_1\times \cdots \times \Omega_d=:\Omega\subseteq \RR^d$ be a vector of input random variables. Each variable $x_j$ follows a probability distribution with density $\rho_j$ defined over the interval $\Omega_j\subseteq \RR$. Throughout this paper except Section~\ref{subsec:dependent}, we assume that all the random variables are independent with each other. Moreover, for simplicity of notation, we write $[1:d]:=\{1,\ldots,d\}$. For a subset $u\subseteq [1:d]$, we denote the complement of $u$ by $-u=[1:d]\setminus u$ and denote the cardinality of $u$ by $|u|$. For vectors $\bsx,\bsy\in \Omega$, we write $\bsx_u=(x_j)_{j\in u}$ and $(\bsx_u,\bsy_{-u})=\bsz$ with $z_j=x_j$ if $j\in u$ and $z_j=y_j$ otherwise. The Cartesian product $\prod_{j\in u}\Omega_j$ is denoted by $\Omega_u$, and the product of density functions $\prod_{j\in u}\rho_j(x_j)$ is denoted by $\rho_u(\bsx_u)$ with an exception for $u=[1:d]$ in which case we simply write $\rho(\bsx)$.

Now let $f\colon \Omega\to \RR$ be a function which outputs a quantity of our interest. If the variance of $f$ with respect to $\bsx$ is finite, $f$ can be decomposed as
\begin{align}\label{eq:anova1} f(\bsx)= \sum_{u\subseteq [1:d]}f_u(\bsx_u), \end{align}
where each summand is recursively defined by
\[ f_{\emptyset} = \int_{\Omega}f(\bsx)\rd \bsx =: \mu\]
and 
\[ f_{u}(\bsx_u) = \int_{\Omega_{-u}}f(\bsx)\rho_{-u}(\bsx_{-u})\rd \bsx_{-u}-\sum_{v\subsetneq u}f_v(\bsx_v) \]
for any non-empty subset $u$.
The following lemma shows important properties of this decomposition.
\begin{lemma}\label{lem:anova}
With the notation above, the following holds true.
\begin{enumerate}
\item For any non-empty subset $u$ and coordinate $j\in u$, we have
\[ \int_{\Omega_j}f_u(\bsx_u)\rho_j(x_j)\rd x_j = 0. \]
\item For any $u,v\subseteq [1:d]$, we have
\[ \int_{\Omega}f_u(\bsx_u)f_v(\bsx_v)\rho(\bsx)\rd \bsx = \begin{cases} \displaystyle \sigma_u^2:=\int_{\Omega_u}\left(f_{u}(\bsx_u)\right)^2 \rho_u(\bsx_u)\rd \bsx_u & \text{if $u=v$,}\\ 0 & \text{otherwise.}\end{cases} \]
\end{enumerate}
\end{lemma}
\begin{proof}
We refer the reader to \cite[Appendix~A.3]{Owenbook} for the proof of the special case where $\rho_j$ is the uniform distribution over the unit interval $[0,1]$, which can be easily generalized to a proof of this lemma.
\end{proof}

It follows from the second assertion of Lemma~\ref{lem:anova} that the variance of $f$ can be decomposed as
\begin{align*}
\sigma^2 & := \int_{\Omega}\left(f(\bsx)-\mu\right)^2\rho(\bsx)\rd\bsx = \int_{\Omega}\left(\sum_{\emptyset \neq u\subseteq [1:d]}f_u(\bsx_u)\right)^2 \rho(\bsx)\rd\bsx \\
& \: = \sum_{\emptyset \neq u,v\subseteq [1:d]}\int_{\Omega}f_u(\bsx_u)f_v(\bsx_v) \rho(\bsx)\rd\bsx = \sum_{\emptyset \neq u\subseteq [1:d]}\sigma_u^2.
\end{align*}
This way, the \emph{overall variance} $\sigma^2$ is decomposed into $2^d-1$ terms with each term $\sigma_u^2$ being the variance of a lower-dimensional function $f_u$. This is why we call the decomposition \eqref{eq:anova1} the \emph{analysis of variance (ANOVA) decomposition} of $f$.

\subsection{Sensitivity indices}\label{subsec:index}
For a non-empty subset $u\subseteq [1:d]$, the main effect and the total effect for a group of input variables $\bsx_u$ are defined by
\[ \underline{\tau}^2_u := \sum_{\emptyset \neq v\subseteq u}\sigma_v^2 \quad \text{and}\quad \overline{\tau}^2_u := \sum_{\substack{\emptyset \neq v\subseteq [1:d]\\ v\cap u\neq \emptyset}}\sigma_v^2=\sigma^2-\underline{\tau}^2_{-u}, \]
respectively. We see from the definition that the main effect measures the variance explained by the variables $\bsx_u$. On the other hand, the total effect takes into account all possible interactions with the variables $\bsx_u$, so that we always have $\underline{\tau}^2_u\leq \overline{\tau}^2_u$. We see from the second equality that the total effect measures the difference between the overall variance and the variance explained by the complement variables $\bsx_{-u}$.

The following identities on these effects are well-known:
\[ \underline{\tau}^2_u = \int_{\Omega}\int_{\Omega}f(\bsx)f(\bsx_u,\bsy_{-u})\rho(\bsx)\rho(\bsy)\rd \bsx\rd \bsy -\mu^2, \]
and
\begin{align}\label{eq:identity_total}
 \overline{\tau}^2_u = \frac{1}{2}\int_{\Omega}\int_{\Omega}\left(f(\bsx)-f(\bsy_u,\bsx_{-u})\right)^2\rho(\bsx)\rho(\bsy)\rd \bsx\rd\bsy,
\end{align}
see, e.g., \cite[Appendix~A.6]{Owenbook}. These identities have been used for constructing so-called \emph{pick-freeze} Monte Carlo estimators of the main and total effects, respectively. We refer the reader to  \cite{Sobol01,SAACRT10,Owen13,JKLNP14} among many others.

Looking at these effects for an individual variable, we have
\[ \sum_{j=1}^{d}\underline{\tau}^2_{\{j\}}=\sum_{j=1}^{d}\sigma^2_{\{j\}} \leq \sigma^2, \]
and
\[ \sum_{j=1}^{d}\overline{\tau}^2_{\{j\}} = \sum_{\emptyset \neq u\subseteq [1:d]}|u|\sigma_u^2 \geq \sigma^2. \]
Therefore, the \emph{rescaled} versions of these effects, $\underline{\tau}^2_{\{j\}}/\sigma^2$ and $\overline{\tau}^2_{\{j\}}/\sigma^2$, are not summed up to 1. This normalization issue can be circumvented by employing a new class of sensitivity indices called Shapley effects \cite{Owen14}. For an individual variable $x_j$, the Shapley effect is defined as follows:
\begin{align}\label{eq:shapley_def} \phi_j := \frac{1}{d}\sum_{u\subseteq -\{j\}}\binom{d-1}{|u|}^{-1}\left( \underline{\tau}^2_{u+j}-\underline{\tau}^2_{u}\right) = \frac{1}{d}\sum_{u\subseteq -\{j\}}\binom{d-1}{|u|}^{-1}\left( \overline{\tau}^2_{u+j}-\overline{\tau}^2_{u}\right), \end{align}
where we write $\underline{\tau}^2_{\emptyset}=\overline{\tau}^2_{\emptyset}=0$. Here the second equality is proven in \cite[Theorem~1]{SNS16}. As pointed out in \cite{Owen14},
\[ \sum_{j=1}^{d}\phi_j = \sigma^2 \]
holds, so that the Shapley effect enables an easy interpretation when measuring the relative importance of individual variables. Moreover, it is known from \cite[Theorem~1]{Owen14} that
\[ \phi_j = \sum_{\substack{\emptyset \neq u\subseteq [1:d]\\ j\in u}}\frac{\sigma_u^2}{|u|}, \]
which leads to an inequality $\underline{\tau}^2_{\{j\}}\leq \phi_j\leq \overline{\tau}^2_{\{j\}}$.

\section{Monte Carlo estimator for Shapley effect}\label{sec:mc_estimate}

\subsection{Algorithm}
In this section, we introduce a simple unbiased Monte Carlo algorithm to estimate the Shapley effects $\phi_j$ for all $1\leq j\leq d$. Although we shall use the second expression of $\phi_j$ in \eqref{eq:shapley_def}, it is also possible to construct a similar Monte Carlo algorithm based on the first expression. For a start, let us consider a single value $\phi_j$. Let $\ell\in \{0,1,\ldots,d-1\}$ be a uniformly distributed discrete random variable, and given $\ell$, let $\Ucal_{j}(\ell)\subseteq -\{j\}$ be a uniformly distributed random subset with fixed cardinality $\ell$. Then it follows from \eqref{eq:shapley_def} and \eqref{eq:identity_total} that the Shapley effect for a variable $x_j$ is given by
\begin{align*}
 \phi_j & = \frac{1}{d}\sum_{\ell=0}^{d-1}\binom{d-1}{\ell}^{-1}\sum_{\substack{u\subseteq -\{j\}\\ |u|=\ell}}\left( \overline{\tau}^2_{u+j}-\overline{\tau}^2_{u}\right) \\
 & = \EE_{\ell}\EE_{\Ucal_{j}(\ell)}\left[ \overline{\tau}^2_{\Ucal_{j}(\ell)+j}-\overline{\tau}^2_{\Ucal_{j}(\ell)}\right] \\
 & = \EE_{\ell}\EE_{\Ucal_{j}(\ell)}\left[\int_{\Omega}\int_{\Omega}g_{\Ucal_{j}(\ell)}(\bsx,\bsy)\rho(\bsx)\rho(\bsy)\rd \bsx\rd\bsy\right] \\
 & = \int_{\Omega}\int_{\Omega}\EE_{\ell}\EE_{\Ucal_{j}(\ell)}\left[g_{\Ucal_{j}(\ell)}(\bsx,\bsy)\right]\rho(\bsx)\rho(\bsy) \rd \bsx\rd\bsy,
\end{align*} 
wherein we have defined
\begin{align*}
 g_{\Ucal_{j}(\ell)}(\bsx,\bsy) & = \frac{1}{2}\left(f(\bsx)-f(\bsy_{\Ucal_{j}(\ell)+j},\bsx_{-(\Ucal_{j}(\ell)+j)})\right)^2 - \frac{1}{2}\left(f(\bsx)-f(\bsy_{\Ucal_{j}(\ell)},\bsx_{-\Ucal_{j}(\ell)})\right)^2 \\
 & = \left(f(\bsx)-\frac{f(\bsy_{\Ucal_{j}(\ell)},\bsx_{-\Ucal_{j}(\ell)})+f(\bsy_{\Ucal_{j}(\ell)+j},\bsx_{-(\Ucal_{j}(\ell)+j)})}{2}\right) \\
 & \quad \times \left(f(\bsy_{\Ucal_{j}(\ell)},\bsx_{-\Ucal_{j}(\ell)})-f(\bsy_{\Ucal_{j}(\ell)+j},\bsx_{-(\Ucal_{j}(\ell)+j)})\right).
\end{align*}

This representation naturally leads to the following unbiased Monte Carlo estimator of $\phi_j$:
\[ \widehat{\phi}_{j,N}=\frac{1}{N}\sum_{n=1}^{N}g_{\Ucal^{(n)}_{j}(\ell^{(n)})}(\bsx^{(n)},\bsy^{(n)}), \]
where $\bsx^{(1)},\ldots,\bsx^{(N)}$ and $\bsy^{(1)},\ldots,\bsy^{(N)}$ are i.i.d.\ samples generated from the density $\rho$, $\ell^{(1)},\ldots,\ell^{(N)}$ are i.i.d.\ samples of $\ell\in \{0,1,\ldots,d-1\}$, and for each $n$, $\Ucal^{(n)}_{j}(\ell^{(n)})$ denotes a random sample of $\Ucal_{j}(\ell^{(n)})$ conditional on $\ell^{(n)}$. 

The key difference from the algorithms proposed in \cite{SNS16} and \cite{BBD20} is that we directly apply the pick-freeze scheme to estimate the difference $\overline{\tau}^2_{u+j}-\overline{\tau}^2_{u}$ instead of estimating two terms independently. This idea not only gives a tight coupling in estimating the difference, but also exploits a sensible recycling idea from \cite{SNS16} well further. In fact, if each of $\phi_1,\ldots, \phi_d$ is estimated independently, we need $3dN$ total function evaluations since computing each of $g_{\Ucal^{(n)}_{j}(\ell^{(n)})}(\bsx^{(n)},\bsy^{(n)})$ requires evaluating the function $f$ three times, whereas we can reduce this cost to $(d+1)N$ through the recycling technique by estimating all $\phi_1,\ldots, \phi_d$ simultaneously.

Let $\bspi =(\pi(1),\ldots,\pi(d))$ denote a random permutation of $[1:d]$. For a fixed $1\leq j\leq d$, generating $\ell\in \{0,1,\ldots,d-1\}$ and $\Ucal_{j}(\ell)$ randomly is equivalent to finding $\ell$ such that $\pi(\ell+1)=j$ and set $\Ucal_{j}(\ell)=\{\pi(1),\ldots,\pi(\ell)\}$. Note that we set $\Ucal_{j}(0)$ to the empty set. Because of this equivalence, by writing $\bspi_j=\{\pi(1),\ldots,\pi(\ell)\}$ with $\ell$ satisfying $\pi(\ell+1)=j$, we can rewrite our estimator as
\[ \widehat{\phi}_{j,N}=\frac{1}{N}\sum_{n=1}^{N}g_{\bspi^{(n)}_j}(\bsx^{(n)},\bsy^{(n)}), \]
for i.i.d.\ permutations $\bspi^{(1)},\ldots,\bspi^{(N)}$ of $[1:d]$. The recycling technique from \cite{SNS16} shares the same $\bspi^{(1)},\ldots,\bspi^{(N)}$ for computing all of $\widehat{\phi}_{1,N},\ldots,\widehat{\phi}_{d,N}$, which still ensures the unbiaseness of estimators. After some rearrangements, our algorithm can be summarized as follows.
\begin{algorithm}[Monte Carlo estimation of Shapley effects for all input variables]\label{alg}
Let $d$ be the number of input variables and $N$ be the sample size. Initialize $\widehat{\phi}_{1,N}=\cdots =\widehat{\phi}_{d,N}=0$. For $1\leq n\leq N$, do the following:
\begin{enumerate}
\item Generate $\bsx^{(n)},\bsy^{(n)}\in \Omega$ and $\bspi^{(n)}$ randomly.
\item Let $\ell=1$ and compute $\Fcal_n=\Fcal^{-}_n=f(\bsx^{(n)})$.
\item 
\begin{enumerate}
\item Compute $\Fcal^{+}_n=f(\bsy^{(n)}_{\{\pi(1),\ldots,\pi(\ell)\}}, \bsx^{(n)}_{-\{\pi(1),\ldots,\pi(\ell)\}})$.
\item Update
\[ \widehat{\phi}_{\pi(\ell),N} = \widehat{\phi}_{\pi(\ell),N}+\frac{1}{N}\left( \Fcal_n-\frac{\Fcal^{-}_n+\Fcal^{+}_n}{2}\right)\left( \Fcal^{-}_n-\Fcal^{+}_n\right). \]
\item Let $\ell=\ell+1$. If $\ell\leq d$, let $\Fcal^{-}_n=\Fcal^{+}_n$ and go to Step~(a). 
\end{enumerate}
\end{enumerate}
\end{algorithm}
\noindent It is obvious that we evaluate function values only $d+1$ times for each $n$, leading to the total computational cost of $(d+1)N$.

\begin{remark}
In Algorithm~\ref{alg}, instead of generating $\bsx^{(n)},\bsy^{(n)}$ randomly for all $n$, one can use Latin hypercube sampling or (randomized) quasi-Monte Carlo sampling. However, improving the convergence of the estimated Shapley effects seems quite hard. The reason is that we also need to generate a discrete object $\bspi^{(n)}$ and the integrand $g_{\bspi_j}(\bsx,\bsy)$ changes depending on the samples $\bsx^{(n)}$ and $\bsy^{(n)}$, so that the standard theory of quasi-Monte Carlo sampling which requires a smoothness of integrands does not apply. This problem does not occur for the main and total effects.
\end{remark}
\subsection{Some properties}
Because of the linearity of expectation, the first property on the unbiasedness of our Monte Carlo estimator, mentioned in Section~\ref{sec:intro}, can be easily shown as follows:
\[ \EE\left[\widehat{\phi}_{j,N}\right] =\frac{1}{N}\sum_{n=1}^{N}\EE\left[g_{\Ucal^{(n)}_{j}(\ell^{(n)})}(\bsx^{(n)},\bsy^{(n)})\right] = \EE_{\ell}\EE_{\Ucal_{j}(\ell)}\EE_{\bsx,\bsy\sim \rho}\left[g_{\Ucal_{j}(\ell)}(\bsx,\bsy)\right] = \phi_j.\]
This leads to the fourth property
\[ \EE\left[\sum_{j=1}^{d}\widehat{\phi}_{j,N}\right] = \sum_{j=1}^{d}\EE\left[\widehat{\phi}_{j,N}\right]=\sum_{j=1}^{d}\phi_j=\sigma^2. \]

We now show the second property of the estimator. The third property immediately follows from the central limit theorem. Let us make the second property more explicit.
\begin{theorem}\label{thm1}
Assume that 
\[ M_4(f):=\int_{\Omega}\left( f(\bsx)\right)^4 \rho(\bsx)\rd \bsx < \infty. \]
Then the following holds true.
\begin{enumerate}
\item The variance of $\widehat{\phi}_{j,N}$ is finite and decays at the rate of $1/N$ for all $j$.
\item The variance of $\widehat{\phi}_{j,N}$ is estimated unbiasedly by
\[ \frac{1}{N(N-1)}\sum_{n=1}^{N}\left(g_{\Ucal^{(n)}_{j}(\ell^{(n)})}(\bsx^{(n)},\bsy^{(n)})-\widehat{\phi}_{j,N}\right)^2. \]
\end{enumerate}
\end{theorem}
\begin{proof}
Since the second assertion is well-known \cite[Chapter~2]{Owenbook}, we only give a proof for the first assertion. Because of the independence between different samples, we have
\begin{align*}
\VV\left[ \widehat{\phi}_{j,N}\right] = \frac{\VV\left[ g_{\Ucal_{j}(\ell)}(\bsx,\bsy)\right]}{N},
\end{align*}
where the variance on the right-hand side is taken with respect to $\bsx,\bsy\sim \rho$, $\ell\in \{0,1,\ldots,d-1\}$ and $\Ucal_{j}(\ell)\subseteq -\{j\}$. Thus it suffices to prove that $\VV\left[ g_{\Ucal_{j}(\ell)}(\bsx,\bsy)\right]$ is finite under the assumption $M_4(f)<\infty$. In fact, applying Jensen's inequality twice, we have
\begin{align*}
& \VV\left[ g_{\Ucal_{j}(\ell)}(\bsx,\bsy)\right] \leq \EE\left[ \left( g_{\Ucal_{j}(\ell)}(\bsx,\bsy) \right)^2\right] \\
& = \EE_{\ell}\EE_{\Ucal_{j}(\ell)}\left[ \int_{\Omega}\int_{\Omega}\left(g_{\Ucal_{j}(\ell)}(\bsx,\bsy)\right)^2 \rho(\bsx)\rho(\bsy)\rd \bsx\rd\bsy\right] \\
& = \frac{1}{d}\sum_{\ell=0}^{d-1}\binom{d-1}{\ell}^{-1}\sum_{\substack{u\subseteq -\{j\}\\ |u|=\ell}} \int_{\Omega}\int_{\Omega}\rho(\bsx)\rho(\bsy) \\
& \quad \quad \times \left(\frac{1}{2}\left(f(\bsx)-f(\bsy_{u+j},\bsx_{-(u+j)})\right)^2 - \frac{1}{2}\left(f(\bsx)-f(\bsy_{u},\bsx_{-u})\right)^2 \right)^2 \rd \bsx\rd\bsy \\
& \leq \frac{1}{d}\sum_{\ell=0}^{d-1}\binom{d-1}{\ell}^{-1}\sum_{\substack{u\subseteq -\{j\}\\ |u|=\ell}} \int_{\Omega}\int_{\Omega}\rho(\bsx)\rho(\bsy) \\
& \quad \quad \times \left(\frac{1}{2}\left(f(\bsx)-f(\bsy_{u+j},\bsx_{-(u+j)})\right)^4 + \frac{1}{2}\left(f(\bsx)-f(\bsy_{u},\bsx_{-u})\right)^4 \right) \rd \bsx\rd\bsy \\
& \leq \frac{1}{d}\sum_{\ell=0}^{d-1}\binom{d-1}{\ell}^{-1}\sum_{\substack{u\subseteq -\{j\}\\ |u|=\ell}} \int_{\Omega}\int_{\Omega}\rho(\bsx)\rho(\bsy) \\
& \quad \quad \times \left(8\left(f(\bsx)\right)^4+4\left(f(\bsy_{u+j},\bsx_{-(u+j)})\right)^4 + 4\left(f(\bsy_{u},\bsx_{-u})\right)^4 \right) \rd \bsx\rd\bsy \\
& = 16 M_4(f) < \infty.
\end{align*}
Thus we are done.
\end{proof}

\section{Numerical experiments}\label{sec:numerics}
We conduct some numerical experiments to confirm the effectiveness of our Shapley effect estimator and compare the result with the {\tt shapleyPermRand} estimator proposed in \cite{SNS16}, which is implemented in the {\tt sensitivity} package of the {\tt R} software \cite{R_sensitivity}. For the first test case, we consider Ishigami function which is a function of only three variables but is interesting in terms of interaction and non-linearity. For the second test case, we consider Sobol' $g$ function which has been often used as a standard test problem in the context of global sensitivity analysis in high dimensions \cite{GSAPrimer}. The third example is taken from \cite[Section~7]{SZ16} which involves the assessment of structural component strength. A short implementation of our algorithm in Matlab for the second test case is available in \ref{app:matlab}.

\subsection{Ishigami function}
Let us consider the case where $\rho$ is the uniform distribution over the domain $\Omega=[-\pi,\pi]^3$. Ishigami function is defined by
\[ f(x_1,x_2,x_3) = \left(1+bx_3^4\right)\sin x_1+a\left( \sin x_2\right)^2, \]
with some coefficients $a,b>0$. For this function, we see that
\begin{align*}
& \underline{\tau}^2_{\{1\}}=\frac{1}{2}\left(1+\frac{\pi^4b}{5}\right)^2,\quad \underline{\tau}^2_{\{2\}}=\frac{a^2}{8},\quad \underline{\tau}^2_{\{3\}}=0,\\
& \overline{\tau}^2_{\{1\}}=\underline{\tau}^2_{\{1\}}+\frac{8\pi^8 b^2}{225},\quad \overline{\tau}^2_{\{2\}}=\underline{\tau}^2_{\{2\}},\quad \overline{\tau}^2_{\{3\}}=\frac{8\pi^8 b^2}{225}, \\
& \phi_1=\underline{\tau}^2_{\{1\}}+\frac{4\pi^8 b^2}{225},\quad \phi_2=\underline{\tau}^2_{\{2\}},\quad \phi_3=\frac{4\pi^8 b^2}{225} .
\end{align*}
This means that the second variable $x_2$ does not have any interaction with the other variables and that the third variable $x_3$ affects the variation of $f$ only through an interaction with the first variable $x_1$.

In what follows, we choose $a=7$ and $b=0.1$. The Shapley effects for all the input variables are estimated according to Algorithm~\ref{alg}. As a comparison, they are also estimated by the {\tt shapleyPermRand} estimator with the choice of $N_v=2N$, $m=\lceil N/3\rceil$, $N_o=1$ and $N_i=3$ such that the number of the total function evaluations is almost equal to that of our proposed algorithm, which is $(d+1)N$. Although we refer the reader to \cite{R_sensitivity} for what each input parameter means, the choice ($m=\lceil N/3\rceil$, $N_o=1$ and $N_i=3$) is nearly optimal according to the theoretical analysis given in \cite[Claim~2]{SNS16}.

Figure~\ref{fig:ishigami_n14} shows the result for the case $N=2^{14}$. Not only the values of Shapley effects estimated by the two algorithms but also the exact values of the overall variance, main effects, Shapley effects and total effects are plotted. Besides, the 95\% confidence intervals for the Shapley effects are shown by error bars. It is clear that our proposed algorithm provides better estimates of the Shapley effects than the compared algorithm. In particular, the compared algorithm underestimates the Shapley effect for the second variable $x_2$ and the upper confidence limit is slightly larger than the exact value. Moreover, the confidence intervals for our proposed algorithm are narrower than those for the compared algorithm for all the input variables.

\begin{figure}[t]
    \centering
    \includegraphics[width=0.6\textwidth]{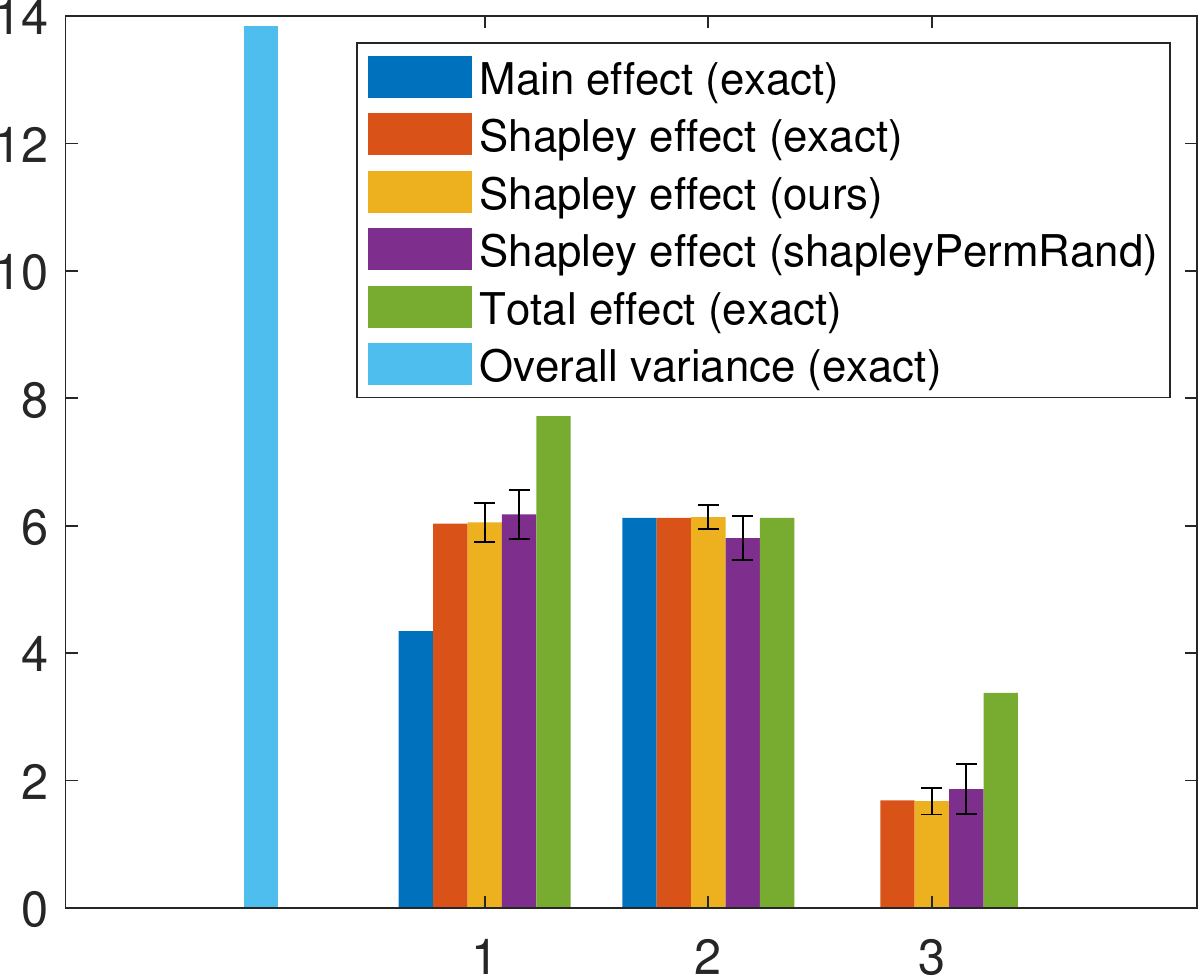}
    \caption{The estimated Shapley effects of individual variables for Ishigami function with $N=2^{14}$. The results obtained by our proposed algorithm are shown in yellow together with the confidence intervals, while those by the {\tt shapleyPermRand} estimator are in purple. The exact overall variance, main effects, Shapley effects and total effects are also shown for comparison.}
    \label{fig:ishigami_n14}
\end{figure}

In order to evaluate the errors from the exact values quantitatively, we run $R=10$ independent trials with a fixed $N$ and estimate the expected sum of squared errors (SSE) by
\begin{align}\label{eq:sse} 
\EE\left[ \sum_{j=1}^{d}\left(\widehat{\phi}_{j,N}-\phi_j\right)^2\right] \approx \frac{1}{R}\sum_{r=1}^{R}\sum_{j=1}^{d}\left(\widehat{\phi}_{j,N}^{(r)}-\phi_j\right)^2,
\end{align}
where $\widehat{\phi}_{j,N}^{(r)}$ denotes the estimated value of $\phi_j$ for the $r$-th trial. Figure~\ref{fig:ishigami_sse} compares the expected SSEs computed for the two algorithms with various values of $N$. The expected SSE for our proposed algorithm is smaller than that for the compared algorithm, which is consistent with the observation of Figure~\ref{fig:ishigami_n14} that the confidence interval for our proposed algorithm is narrower. The expected SSE decays at the rate almost of $1/N$, which supports the theoretical result given in Theorem~\ref{thm1}.

\begin{figure}[t]
    \centering
    \includegraphics[width=0.6\textwidth]{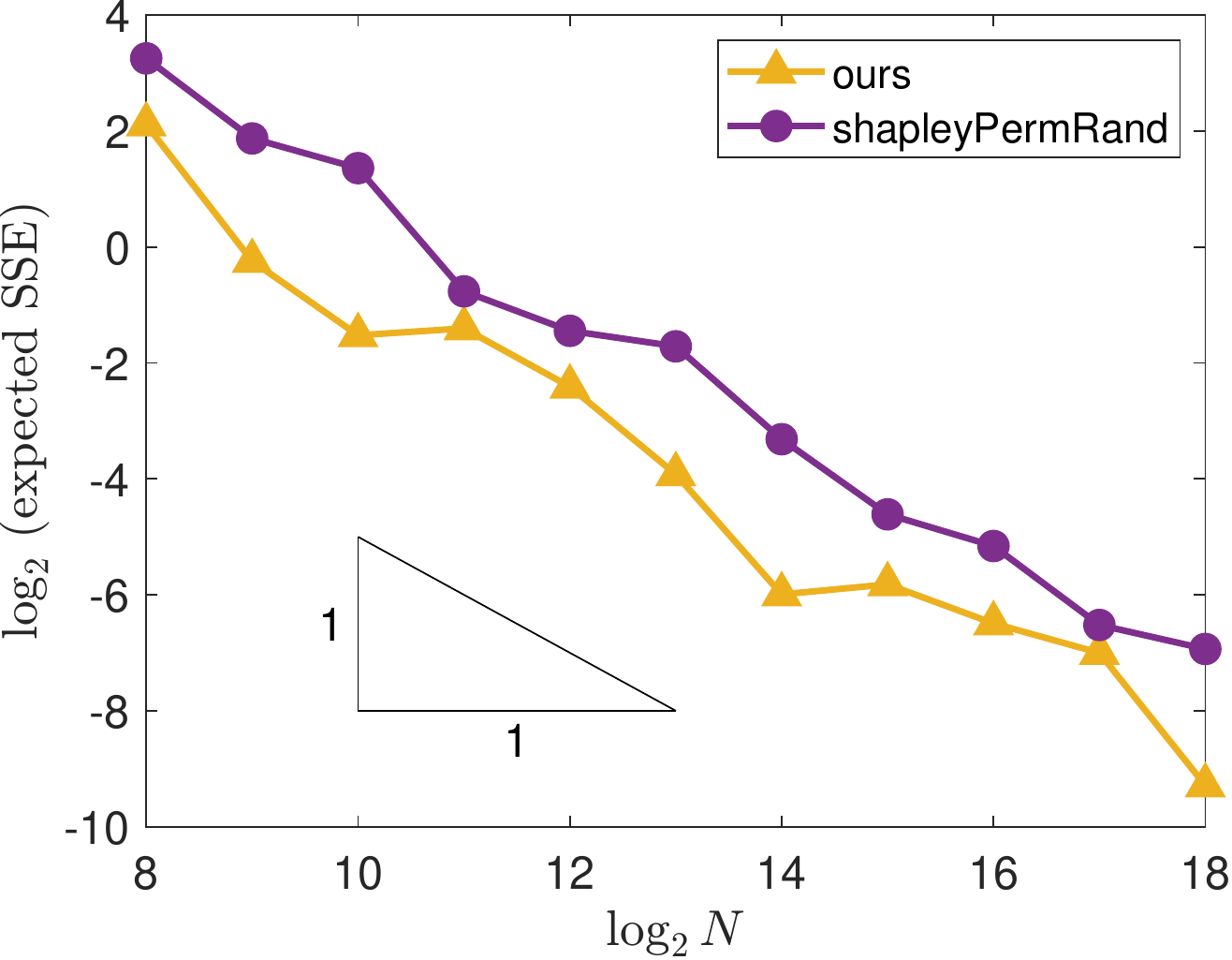}
    \caption{The expected sum of squared errors for Ishigami function as a function of $N$.}
    \label{fig:ishigami_sse}
\end{figure}

\subsection{Sobol' $g$ function}
As a high-dimensional example, let us consider the case where $\rho$ is the uniform distribution over the domain $\Omega=[0,1]^d$ for large $d$. With its simplest form, Sobol' $g$ function is defined by
\[ f(\bsx) = \prod_{j=1}^{d}\frac{|4x_j-2|+a_j}{1+a_j}, \]
with non-negative weight parameters $a_1,\ldots,a_d$. For this function, we can easily obtain
\[ \sigma^2 = \prod_{j=1}^{d}\left[ 1+\frac{1}{3(1+a_j)^2}\right]-1 \quad \text{and}\quad \sigma_u^2 = \prod_{j\in u}\frac{1}{3(1+a_j)^2}, \]
for any non-empty $u\subseteq [1:d]$. Thus, for an individual variable $x_j$, the main and total effects are given by
\[ \underline{\tau}^2_{\{j\}} = \frac{1}{3(1+a_j)^2}\quad \text{and}\quad \overline{\tau}^2_{\{j\}}=\frac{1}{3(1+a_j)^2}\times \prod_{\substack{\ell=1\\ \ell\neq j}}^{d}\left[ 1+\frac{1}{3(1+a_{\ell})^2}\right]\]
respectively. The Shapley effect is 
\[ \phi_j = \sum_{\substack{\emptyset \neq u\subseteq [1:d]\\ j\in u}}\frac{1}{|u|}\prod_{\ell\in u}\frac{1}{3(1+a_{\ell})^2}, \]
which seems hard to simplify further. 

\begin{figure}[t]
    \centering
    \includegraphics[width=0.6\textwidth]{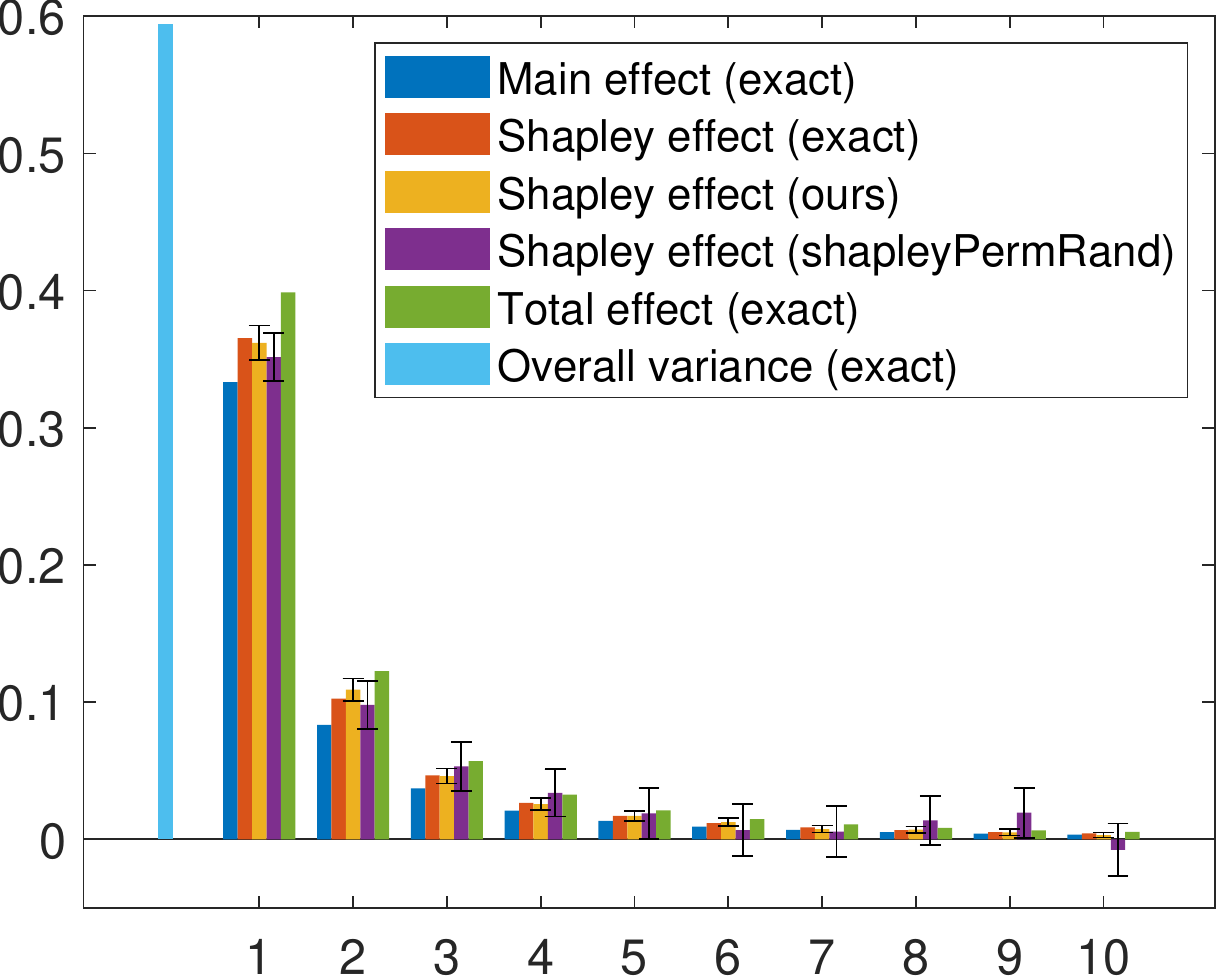}
    \caption{The estimated Shapley effects of individual variables for Sobol' $g$ function with $N=2^{14}$. The results obtained by our proposed algorithm are shown in yellow together with the confidence intervals, while those by the {\tt shapleyPermRand} estimator are in purple. The exact overall variance, main effects, Shapley effects and total effects are also shown for comparison.}
    \label{fig:sobolg_n14}
\end{figure}

In what follows, we set $d=10$ and $a_j=j-1$. With this choice of $a_j$, the relative importance of individual variables is given by the ascending order, i.e., $x_1$ is the most important, $x_2$ is the second most important, and so on. The Shapley effects for all the input variables are estimated according to Algorithm~\ref{alg} and the {\tt shapleyPermRand} estimator with the same choice of $N_v,m,N_o,N_i$ as considered in the first test case. The result with $N=2^{14}$ is shown in Figure~\ref{fig:sobolg_n14}. Similarly to Figure~\ref{fig:ishigami_n14}, we also plot the exact values of the overall variance, main effects, Shapley effects and total effects. We can see that all of the exact sensitivity effects decrease as the index increases. Here again, our proposed algorithm provides better estimates with narrower confidence intervals of the Shapley effects than the compared algorithm. The Shapley effect for the variable $x_{10}$, the least important variable, is estimated to be negative by the compared algorithm, although it must take a non-negative value by definition. Note that our proposed algorithm does not ensure non-negativity, so that such an erratic result may be obtained probabilistically particularly for small $N$. However, it can be confirmed from our result with $N=2^{14}$ that the estimated Shapley effects take the values between the corresponding main and total effects, which agrees with the theory. 

An interesting observation is that the confidence interval for our proposed algorithm gets narrower as the index increases, whereas it stays almost the same width for the compared algorithm. The crucial difference between the two algorithms is that our proposed algorithm directly estimates the difference $\overline{\tau}^2_{u+j}-\overline{\tau}^2_{u}$ instead of estimating $\overline{\tau}^2_{u}$ and $\overline{\tau}^2_{u+j}$ independently. If the variance of the function $[(f(\bsx)-f(\bsy_{u+j},\bsx_{-(u+j)}))^2-(f(\bsx)-f(\bsy_u,\bsx_{-u}))^2]/2$ is smaller than those of $(f(\bsx)-f(\bsy_{u+j},\bsx_{-(u+j)}))^2/2$ and $(f(\bsx)-f(\bsy_u,\bsx_{-u}))^2/2$, a direct estimation of the difference $\overline{\tau}^2_{u+j}-\overline{\tau}^2_{u}$ should be more accurate. This way our proposed algorithm avoids an unnecessary increment of the variance of the estimator, particularly for the input variables with small Shapley effects.

The expected SSEs are estimated by \eqref{eq:sse} with $R=10$ for the two algorithms. Figure~\ref{fig:sobolg_sse} compares the results with various values of $N$. The expected SSE for our proposed algorithm is one order of magnitude smaller than that for the compared algorithm and decays at the rate almost of $1/N$.

\begin{figure}[t]
\centering
    \includegraphics[width=0.6\textwidth]{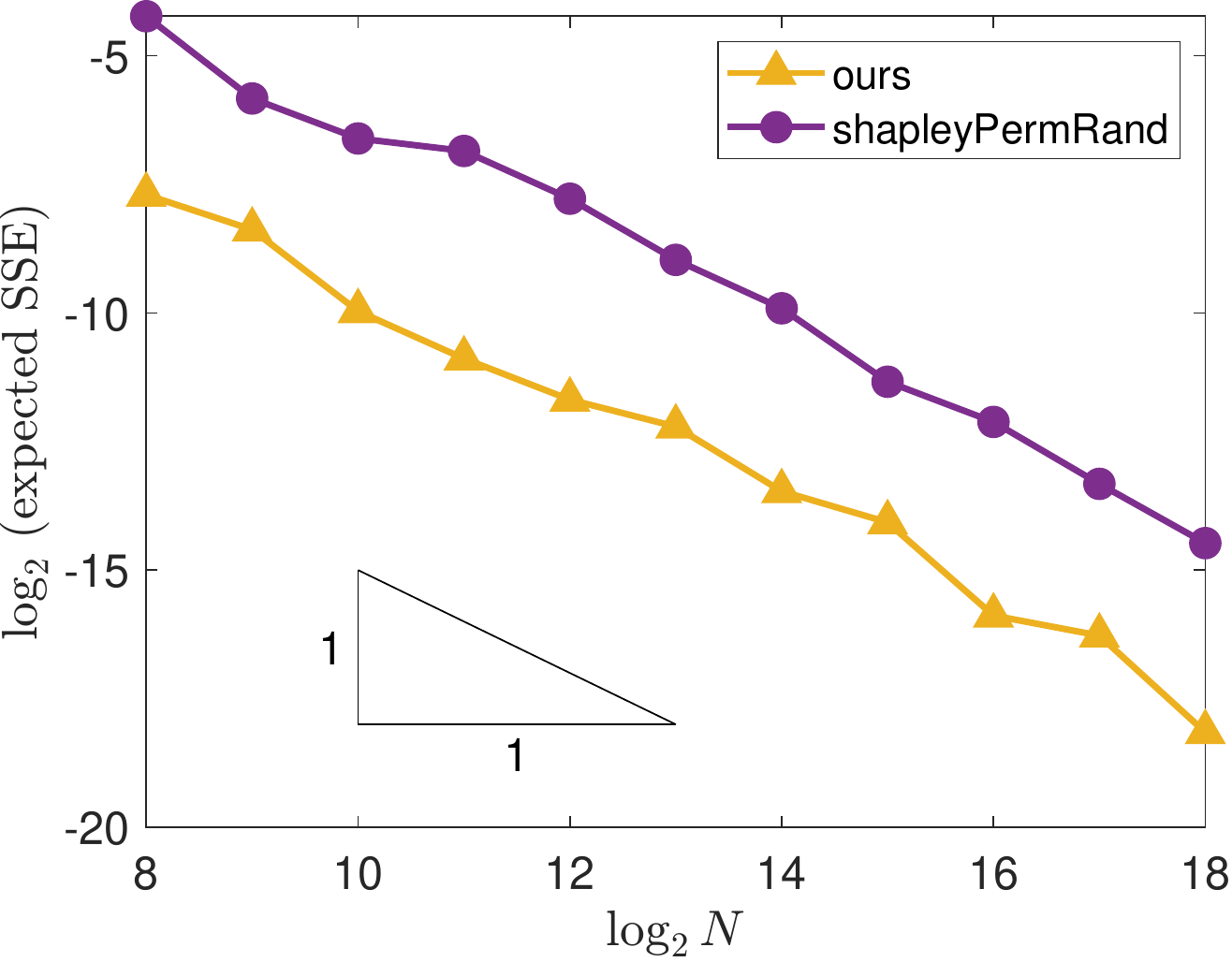}
\caption{The expected sum of squared errors for Sobol' $g$ function as a function of $N$.}
\label{fig:sobolg_sse}
\end{figure}

\subsection{Plate buckling}\label{subsec:plate}
\begin{table}[b]
\caption{List of input variables for plate buckling example. CV denotes the coefficient of variation.}\label{tbl:plate}
\centering
	\begin{tabular}{|l|l|l|l|l|}
		\hline
		variable & description & mean & CV & distribution type\\
		\hline \hline
		$x_1$ & width & 23.808 & 0.028 & normal \\ \hline
		$x_2$ & thickness & 0.525 & 0.044 & log-normal \\ \hline
		$x_3$ & yield stress & 44.2 & 0.1235 & log-normal \\ \hline
		$x_4$ & elastic modulus & 28623 & 0.076 & normal \\ \hline
		$x_5$ & initial deflection & 0.35 & 0.05 & normal \\ \hline
		$x_6$ & residual stress & 5.25 & 0.07 & normal \\ \hline
	\end{tabular}
\end{table}

Here we consider a more realistic example from structural engineering. As explained in \cite{SZ16}, let us consider the buckling strength of a rectangular plate that is supported on all four edges subjected to uniaxial compression. As described in Table~\ref{tbl:plate}, we have $d=6$ input random variables $x_1,\ldots,x_6$ which are related to the material, geometry and imperfection of the plate. The buckling strength is defined as a function of these variables, and is explicitly given by
\[ f(\bsx) = \left( \frac{2.1}{\lambda}-\frac{0.9}{\lambda^2}\right)\left(1-\frac{0.75x_5}{\lambda} \right)\left( 1-\frac{2x_2x_6}{x_1}\right)\qquad \text{with}\quad \lambda = \frac{x_1}{x_2}\sqrt{\frac{x_3}{x_4}}. \]

Since closed formulas for main and total effects are no longer available for this example, we estimate them for all the input variables by using the estimators proposed in the literature.  We use the one from \cite{Owen13} for the main effects, whereas the standard one based on the identity \eqref{eq:identity_total} is used for the total effects, see, e.g., \cite{SAACRT10}. The sample size is set large enough ($N=2^{24}$) to make sure that the estimates are converged sufficiently. The Shapley effects for all the input variables are estimated according to Algorithm~\ref{alg} and the {\tt shapleyPermRand} estimator with the same choice of $N_v,m,N_o,N_i$ as before. The results are shown in Figure~\ref{fig:sz_n14}. Interestingly, there is no clear difference between the main and total effects for all the input variables. This means that each variable does not have significant interactions with other variables for the buckling strength, so that the buckling strength can be effectively written as a sum of $d$ one-dimensional functions $f_{\{1\}},\ldots,f_{\{d\}}$.

\begin{figure}[t]
    \centering
    \includegraphics[width=0.6\textwidth]{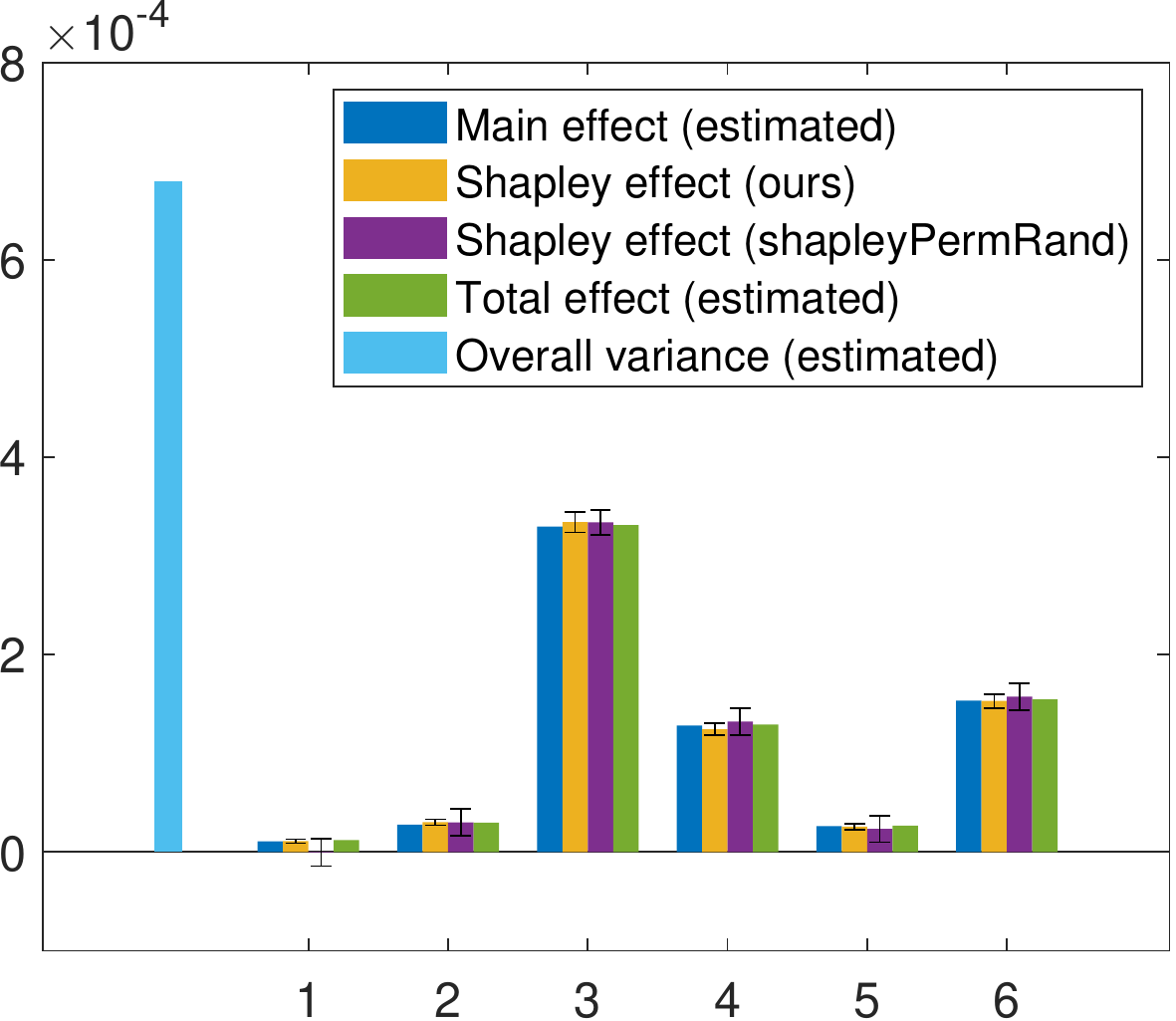}
    \caption{The estimated Shapley effects of individual variables for the plate buckling example with $N=2^{14}$. The results obtained by our proposed algorithm are shown in yellow together with the confidence intervals, while those by the {\tt shapleyPermRand} estimator are in purple. The estimated overall variance, main effects, and total effects are also shown for comparison.}
    \label{fig:sz_n14}
\end{figure}

As already pointed out, the Shapley effect takes its value between the corresponding main and the total effects for each input variable. Since the difference between the latter two values is quite small in this example, it might be challenging to estimate the Shapley effect in a way that the resulting value is bracketed by them. However, our proposed algorithm with $N=2^{14}$ already provides good estimates for all the variables. Similarly to the second test case, the confidence interval for our proposed algorithm is narrower for a variable with smaller Shapley effect, while it is not the case for the compared algorithm.

For this example, the expected SSE is estimated by
\[ \frac{1}{R-1}\sum_{r=1}^{R}\sum_{j=1}^{d}\left(\widehat{\phi}_{j,N}^{(r)}-\overline{\phi}_j\right)^2\quad \text{with}\quad \overline{\phi}_j=\frac{1}{R}\sum_{r=1}^{R}\widehat{\phi}_{j,N}^{(r)}.\]
Again we choose $R=10$. The results with various values of $N$ are compared in Figure~\ref{fig:sz_sse}. We can confirm that the expected SSE for our proposed algorithm is always smaller than that for the compared algorithm and decays at the rate almost of $1/N$.

\begin{figure}[t]
    \centering
    \includegraphics[width=0.6\textwidth]{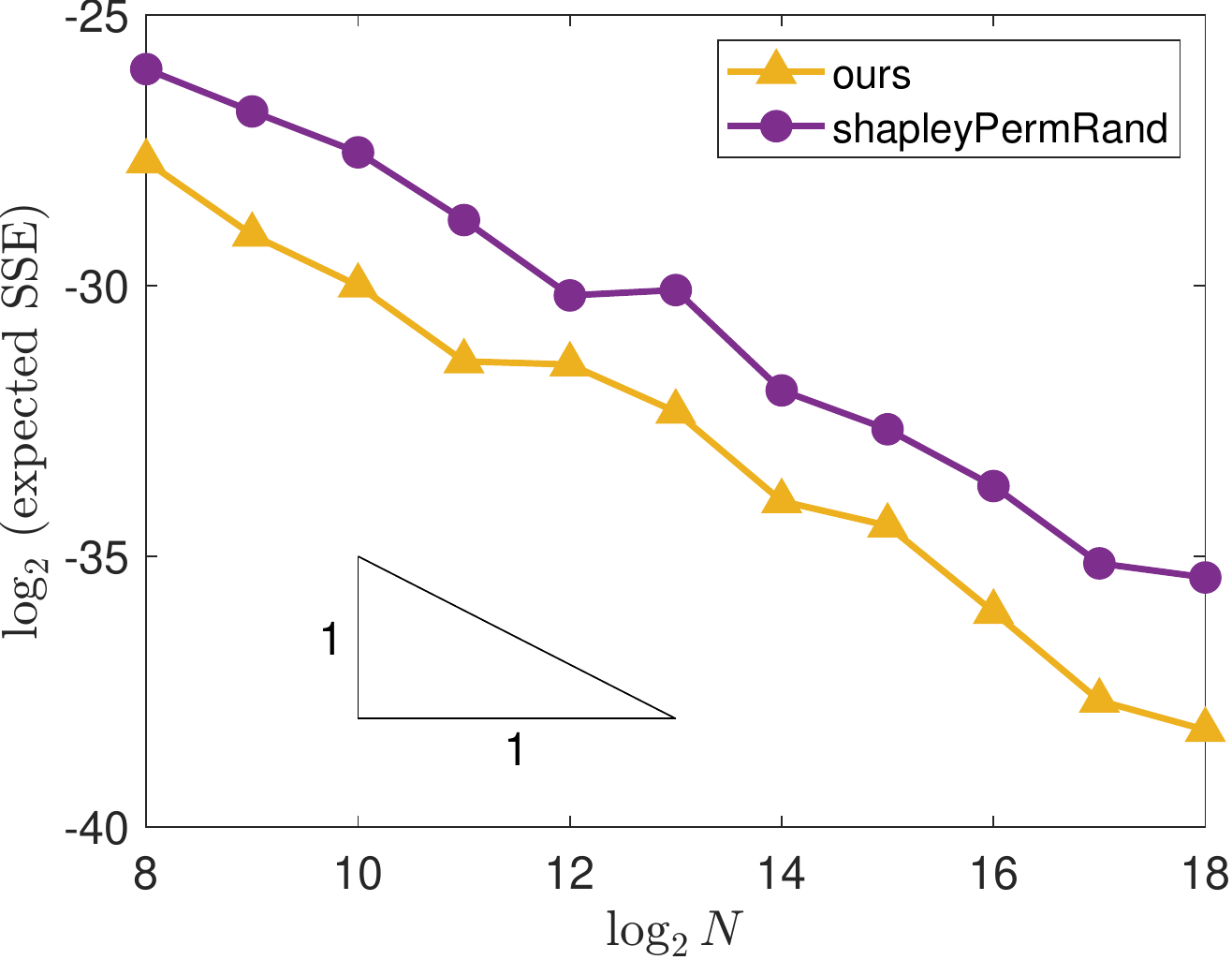}
    \caption{The expected sum of squared errors for the plate buckling example as a function of $N$.}
    \label{fig:sz_sse}
\end{figure}

\section{Possible extensions}\label{sec:extension}
Finally we discuss two possible extensions of our Monte Carlo estimator. The first extension allows for an additional cost reduction from $(d+1)N$ to $dN+1$ or $dN$, which might be useful when $d$ is small but $N$ is large. The second extension considers the case where the input variables are dependent.

\subsection{Additional cost saving}\label{subsec:saving}
In Algorithm~\ref{alg}, we generate two independent sequences for $\bsx$ and $\bsy$, and consider a pairing $(\bsx^{(n)},\bsy^{(n)})$ for $1\leq n\leq N$. Because of the linearity of expectation, however, the unbiasedness of our Monte Carlo estimator still holds even by generating only one random sequence for $\bsx$ and then considering a pairing $(\bsx^{(n)},\bsx^{(n+1)})$ for $1\leq n\leq N$ instead. This strategy is called the \emph{winding star} sampling \cite{JRD94} and enables to reduce the necessary cost from $(d+1)N$ to $dN+1$. The cost of $dN$ is possible by setting $\bsx^{(N+1)}=\bsx^{(1)}$, although the algorithm then becomes not extensible in $N$.

The variance of this lower-cost estimator is finite and decays at the $1/N$ rate under the same assumption given in Theorem~\ref{thm1}. However, because of the dependence between consecutive samples, an unbiased estimation of the variance of the estimator is not possible in the same way as the second item of Theorem~\ref{thm1}.

\subsection{Dependent input case}\label{subsec:dependent}
For $\bsx,\bsy\sim \rho$ and any subset $u\subseteq [1:d]$, it holds that $(\bsx_u,\bsy_{-u})\sim \rho$ as long as the input variables are independent. This property plays a fundamental role in establishing the simple concatenation approach for two independent samples $\bsx^{(n)}, \bsy^{(n)}$ in Algorithm~\ref{alg}. In fact, if the input variables are dependent, we generally have $(\bsx_u,\bsy_{-u})\not\sim \rho$, which prevents us from extending Algorithm~\ref{alg} to the dependent input case in a straightforward manner.

In order to adapt Algorithm~\ref{alg} to the dependent input variables, we need the following modifications:
\begin{enumerate}
\item In the first step, generate only $\bsx^{(n)}\in \Omega$ and $\bspi^{(n)}$ randomly.
\item In the third step, generate $\bsy^{(n)}_{\{\pi(1),\ldots,\pi(\ell)\}}$ conditional only on $\bsx^{(n)}_{-\{\pi(1),\ldots,\pi(\ell)\}}$ and independently of the existing ones $\bsy^{(n)}_{\{\pi(1)\}},\ldots,\bsy^{(n)}_{\{\pi(1),\ldots,\pi(\ell-1)\}}$ for all $1\leq \ell\leq d$.
\end{enumerate}
This way, we still have an unbiased Monte Carlo estimator for the Shapley effects with the cost of $(d+1)N$. However, the independence between $\bsy^{(n)}_{\{\pi(1),\ldots,\pi(\ell-1)\}}$ and $\bsy^{(n)}_{\{\pi(1),\ldots,\pi(\ell)\}}$ may lead to a large variance in estimating the difference $\overline{\tau}^2_{\{\pi(1),\ldots,\pi(\ell)\}}-\overline{\tau}^2_{\{\pi(1),\ldots,\pi(\ell-1)\}}$. Further detailed investigation is needed to address this issue.

\section*{Funding}
The work of T.G.\ was supported by JSPS KAKENHI Grant Number 20K03744.

\section*{Acknowledgments}
The author would like to thank the associate editor and the referees for helpful suggestions and remarks. He also thanks Elmar Plischke for pointing out the relevant literature and suggesting a way to simplify the Matlab code.


\appendix
\section{Matlab implementation}\label{app:matlab}
{\small
\begin{verbatim}
n = 2^14;                    % sample size
d = 10;                      % dimension
x = rand(n,d);               % Monte Carlo samples for x
y = rand(n,d);               % Monte Carlo samples for y
[~, pm] = sort(rand(n,d),2); % random permutation matrix

func = @(x,a)prod((abs(4.*x-2)+a)./(1+a),2); % Sobol g function
a = (0:d-1);

z = x;
fz1 = func(z,a);
fx = fz1;

phi1 = zeros(1,d); % initialization
phi2 = zeros(1,d); % initialization
for j=1:d
    ind = bsxfun(@eq,pm(:,j),1:d); % compare j-th column with 1:d
    z(ind) = y(ind);
    fz2 = func(z,a);
    fmarg = ((fx-fz1/2-fz2/2).*(fz1-fz2))';
    phi1 = phi1 + fmarg*ind/n;
    phi2 = phi2 + fmarg.^2*ind/n;
    fz1 = fz2;
end

s_all = sum(phi1);            % variance of function
phi2 = (phi2-phi1.^2)./(n-1); % variance of Shapley estimates

disp(phi1);                  % Shapley estimates
disp(phi1-1.96*sqrt(phi2));  % lower confidence limit
disp(phi1+1.96*sqrt(phi2));  % upper confidence limit
\end{verbatim}}


\begin{thebibliography}{99}
\bibitem{GSAPrimer} A. Saltelli, M. Ratto, T. Andres, F. Campolongo, J. Cariboni, D. Gatelli, M. Saisana and S. Tarantola, \emph{Global Sensitivity Analysis. The Primer}, John Wiley \& Sons, Ltd, 2008.
\bibitem{IL15} B. Iooss and P. Lema\^{i}tre, A review on global sensitivity analysis methods, In: G. Dellino and C. Meloni (eds.) Uncertainty Management in Simulation-Optimization of Complex Systems. Springer, Boston, 2015, pp.~101--122.
\bibitem{HS96} T. Homma and A. Saltelli, Importance measures in global sensitivity analysis of nonlinear models, Reliab. Eng. Syst. Safe. 52 (1996), 1--17.
\bibitem{BATS03} E. Borgonovo, G.~E. Apostolakis, S. Tarantola and A. Saltelli, Comparison of global sensitivity analysis techniques and importance measures in PSA, Reliab. Eng. Syst. Safe. 79 (2003), 175--185.
\bibitem{HJSS06} J.~C. Helton, J.~D. Johnson, C.~J. Sallaberry and C.~B. Storlie, Survey of sampling-based methods for uncertainty and sensitivity analysis, Reliab. Eng. Syst. Safe. 91 (2006), 1175--1209.
\bibitem{Sobol93} I.~M. Sobol', Sensitivity estimates for nonlinear mathematical models, Math. Model. Comput. Exper. 1 (1993), 407--414. 
\bibitem{Sobol01} I.~M. Sobol', Global sensitivity indices for nonlinear mathematical models and their Monte Carlo estimates, Math. Comput. Simulation 55 (2001), 271--280. 
\bibitem{Owen14} A.~B. Owen, Sobol' indices and Shapley values, SIAM/ASA J Uncertainty Quantification 2 (2014), 245--251. 
\bibitem{SNS16} E. Song, B.~L. Nelson and J. Staum, Shapley effects for global sensitivity analysis: theory and computation, SIAM/ASA J Uncertainty Quantification 4 (2016), 1060--1083. 
\bibitem{OP17} A.~B. Owen and C. Prieur, On Shapley value for measuring importance of dependent inputs, SIAM/ASA J Uncertainty Quantification 5 (2017), 986--1002. 
\bibitem{BBD20} B. Broto, F. Bachoc and M. Depecker, Variance reduction for estimation of Shapley effects and adaptation to unknown input distribution, SIAM/ASA J Uncertainty Quantification 8 (2020), 693--716. 
\bibitem{CGT09} J. Castro, D. G\'{o}mez and J. Tejada, Polynomial calculation of the Shapley value based on sampling, Comput. Oper. Res. 36 (2009), 1726--1730.
\bibitem{Saltelli02} A. Saltelli, Making best use of model evaluations to compute sensitivity indices, Comput. Phys. Comm. 145 (2002), 280--297.
\bibitem{SAACRT10} A. Saltelli, P. Annoni, I. Azzini, F. Campolongo, M. Ratto and S. Tarantola, Variance based sensitivity analysis of model output. Design and estimator for the total sensitivity index, Comput. Phys. Comm. 181 (2010), 259--270.
\bibitem{Owen13} A.~B. Owen, Better estimation of small Sobol' sensitivity indices, ACM Trans. Model. Comput. Simul. 23 (2013), Article~11. 
\bibitem{JKLNP14} A. Janon, T. Klein, A. Lagnoux, M. Nodet and C. Prieur, Asymptotic normality and efficiency of two Sobol index estimators, ESAIM: PS 18 (2014), 342--364. 
\bibitem{Goda17} T. Goda, Computing the variance of a conditional expectation via non-nested Monte Carlo, Oper. Res. Lett. 45 (2017), 63--67. 
\bibitem{TGM06} S. Tarantola, D. Gatelli and T.~A. Mara, Random balance designs for the estimation of first order global sensitivity indices, Reliab. Eng. Syst. Safe. 91 (2006), 717--727.
\bibitem{S08} B. Sudret, Global sensitivity analysis using polynomial chaos expansions, Reliab. Eng. Syst. Safe. 93 (2008), 964--979.
\bibitem{P10} E. Plischke, An effective algorithm for computing global sensitivity indices (EASI), Reliab. Eng. Syst. Safe. 95 (2010), 354--360.
\bibitem{KS16} K. Konakli and B. Sudret, Global sensitivity analysis using low-rank tensor approximations, Reliab. Eng. Syst. Safe. 156 (2016), 64--83.
\bibitem{ALP21} A. Antoniadis, S. Lambert-Lacroix and J.-M. Poggi, Random forests for global sensitivity analysis: A selective review, Reliab. Eng. Syst. Safe. 206 (2021), 107312.
\bibitem{PT17} C. Prieur and S. Tarantola, Variance-based sensitivity analysis: Theory and estimation algorithms, In: R. Ghanem, D. Higdon and H. Owhadi (eds.) Handbook of Uncertainty Quantification. Springer, Switzerland, 2017, pp.~1217--1239.
\bibitem{Owenbook} A.~B. Owen, \emph{Monte Carlo theory, methods and examples}, \url{http://statweb.stanford.edu/~owen/mc/} Last accessed 1 September 2020.
\bibitem{R_sensitivity} B. Iooss et al., Package `sensitivity', version~1.23.1, 2020.
\bibitem{SZ16} M.~D. Shields and J. Zhang, The generalization of Latin hypercube sampling, Reliab. Eng. Syst. Safe. 148 (2016), 96--108. 
\bibitem{JRD94} M.~J.~W. Jansen, W.~A.~H. Rossing and R.~A. Daamen, Monte Carlo estimation of uncertainty contributions from several independent multivariate sources, In: J. Grasman and G. van Straten (eds.) Predictability and Nonlinear Modelling in Natural Sciences and Economics. Springer, Dordrecht, 1994, pp.~334--343.
\end{thebibliography}
\end{document}